\documentclass[10pt, conference, compsocconf]{IEEEtran}
\usepackage[cmex10]{amsmath}
\usepackage{amssymb}
\usepackage{amsthm}
\usepackage{tikz}
\usetikzlibrary{decorations.markings,decorations.pathreplacing,calc}
\usetikzlibrary{arrows.meta}
\usepackage{enumerate}
\usepackage{comment}
\usepackage{cite}
\usepackage{algorithm}
\usepackage[noend]{algpseudocode}
\usepackage{xspace}
\usepackage{enumitem}
\usepackage{ragged2e}
\usepackage[margin=1in]{geometry}
\usepackage{subfig}
\usepackage[utf8]{inputenc}

\usepackage[colorlinks = true, citecolor=blue]{hyperref}
\usepackage{array}

\newcommand{\be}{\begin{equation}}
\newcommand{\ee}{\end{equation}}
\newcommand{\ba}{\begin{array}}
\newcommand{\ea}{\end{array}}
\newcommand{\bea}{\begin{eqnarray}}
\newcommand{\eea}{\end{eqnarray}}

\newcommand{\NP}{\ensuremath{\mathsf{NP}}\xspace}
\newcommand{\defeq}{:=}

\newcommand{\sharpP}{{\#\mathsf{P}}}
\newcommand{\ceqp}{{\mathsf{C_=P}}}
\newcommand{\poly}{{\mathrm{poly}}}
\usepackage{braket}

\newtheorem{thm}{Theorem}

\newtheorem{cor}{Corollary}
\newtheorem*{cor*}{Corollary}
\newtheorem{lem}{Lemma}
\newtheorem{conj}{Conjecture}
\theoremstyle{definition}
\newtheorem{dfn}{Definition}
\newtheorem{rem}{Remark}

\newcommand{\authnote}[2]{}

\newcolumntype{C}[1]{%
 >{\vbox to 4ex\bgroup\vfill\centering}%
 p{#1}%
 <{\egroup}}  

\begin{document}
\title{Quantum supremacy and hardness of estimating output probabilities of quantum circuits}

\author{
\IEEEauthorblockN{Yasuhiro Kondo, Ryuhei Mori}
\IEEEauthorblockA{School of Computing\\
Tokyo Institute of Technology\\
Tokyo, Japan\\
mori@c.titech.ac.jp}
\and
\IEEEauthorblockN{Ramis Movassagh}
\IEEEauthorblockA{IBM Quantum\\
MIT-IBM Watson AI Lab\\
Cambridge MA, U.S.A.\\
ramis@us.ibm.com}
}

\maketitle

\begin{abstract}
Motivated by the recent experimental demonstrations of quantum supremacy, proving the hardness of the output of random quantum circuits is an imperative near term goal.
We prove under the complexity theoretical assumption of the non-collapse of the polynomial hierarchy that approximating the output probabilities of random quantum
circuits 
to within $\exp(-\Omega(m\log m))$ additive error is hard for any classical computer, where
$m$ is the number of gates in the quantum computation.
More precisely, we show that the above problem is $\#\mathsf{P}$-hard under $\mathsf{BPP}^{\mathsf{NP}}$ reduction.
In the recent
experiments, the quantum circuit has $n$-qubits and the architecture
is a two-dimensional grid of size $\sqrt{n}\times\sqrt{n}$ \cite{arute2019quantum}.
Indeed for constant depth circuits approximating the output
probabilities to within $2^{-\Omega(n\log{n})}$ is hard. For circuits of depth $\log{n}$ or $\sqrt{n}$ for which the anti-concentration
property holds, approximating the output probabilities to within $2^{-\Omega(n\log^2{n})}$ and $2^{-\Omega(n^{3/2}\log n)}$
is hard respectively. We then show that the hardness results extend to any open neighborhood of an arbitrary (fixed) circuit including the trivial circuit with identity gates. We made an effort to find the best proofs and proved these results from first principles, which do not use the standard techniques such as the Berlekamp--Welch algorithm, the usual Paturi's lemma, and Rakhmanov's result.
\end{abstract}

\section{Introduction and related work}
Moore's law for classical (super-)computers is reaching a saturation point because if the computation were done with smaller components confined to smaller spaces, then the quantum effects would become relevant.  Consequently alternative models and architectures are being investigated to empower the future of computation. Among the many proposals, quantum computing is currently the only model of computation that could potentially exponentially outperform any classical computer. Proving this in affirmative has been a main driving force in the field of quantum computation. 

For quantum computers to have the awesome computational power just described the so called Extended Church-Turing Thesis (ECTT) would need to be refuted. ECTT states that a probabilistic Turing machine can efficiently simulate any model of computation that can be realized in Nature (i.e., a realistic computation). A single computational task that would provably refute ECTT would be sufficient. Therefore it is an imperative near-term goal to show that for a given computational task (whatever it may be) a quantum computer can provably outperform any classical computer by running in a time that is exponentially faster. It is then necessary that an actual experiment is performed to demonstrate the separation. Hence to refute the ECTT one needs a solid complexity theoretical foundation and an experimental demonstration. This event (i.e.,~refutation of ECTT) would be a watershed moment in the history of computation, which would usher the era of quantum supremacy.

It is noteworthy that quantum computation has already demonstrated classical ascendancy for search problems, where Grover's algorithm provably gives a quadratic speed-up over the best possible classical search algorithms \cite{grover1996fast}. More striking is the Simons problem which proves that on a quantum computer a hidden sub-string can be found exponentially faster than on a classical computer~\cite{simon1997power}. Its generalization, discovered by Shor, showed that factorization of large composite integers can be done exponentially faster than the best known classical algorithms~\cite{shor1999polynomial}, with significant implications for cybersecurity and cryptography. The inception of quantum computation harks back to Feynman's 1981 paper, which argued that simulation of quantum matter would be exponentially hastened by a quantum computer~\cite{feynman1986quantum}. Grover's algorithm and Simon's algorithm have provable quadratic and exponential speed up, respectively, against classical algorithms in terms of the query complexity. However, the exponential separation of the power of quantum computers over the classical ones in terms of computational complexity has not been proved to date. In fact, the refutation of ECTT remains a major open problem. 
What would be a good task that would establish the exponential separation in the near-term? 

Modern quantum supremacy proposals are based on the hardness of sampling.
The hardness of sampling relies on complexity theoretical assumptions that pre-date quantum computing.
It is known that there is no classical efficient algorithm sampling outputs of the worst-case quantum circuit unless the polynomial hierarchy collapses~\cite{bremner2011classical,terhal2002adaptive}. Related to the sampling problem, the computation of output probability for the worst-case quantum circuit is classically hard~\cite{fenner1998determining,terhal2002adaptive}. 
The first proposal for demonstrating sampling-based quantum supremacy by a near-term quantum computer was the original BosonSampling paper of Aaronson and Arkhipov~\cite{aaronson2011computational} in which they showed that producing samples from a distribution that mimics the distribution of a linear optical system is classically hard. Later, Bremner et al showed that a class of circuits known as IQP circuits are also classically hard to sample from \cite{bremner2016average}. The foremost candidate for demonstrating quantum supremacy has been the so-called Random Circuit Sampling (RCS) problem~\cite{boixo2018characterizing}, which states that for any classical computer it is hard to produce samples from a distribution that is close to the distribution of a local quantum circuit whose local gates are randomly and independently are drawn uniformly from the space of all possible gates. 

Demonstration of quantum supremacy is ultimately given by an experiment for which there is solid complexity theoretical evidence of hardness of the task at hand. Indeed Google did an experiment that involved a random circuit with $53$ qubits to demonstrate the hardness of RCS~\cite{arute2019quantum}. Soon after new classical algorithms emerged that challenged the claim~\cite{napp2019efficient,huang2020classical, pednault2019leveraging}. It remains an open problem to mathematically prove the hardness of sampling. A fruitful approach is to prove the hardness of sampling by proving the hardness of approximating probability amplitudes of the quantum circuit.
In particular, if the probability amplitudes obey an anti-concentration property~\cite{harrow2018approximate,dalzell2020random}, then one can use Stockmeyer's algorithm \cite{stockmeyer1985approximation} to prove that it is sufficient to prove that the amplitudes are hard to approximate to within $2^{-n}/\text{poly}(n)$ additive error. 

The first theoretical evidence for the hardness of computing the output probabilities was given by Bouland et al~\cite{bouland2018quantum}, who showed that the computation of the amplitudes of a {\it non-unitary} approximation of the actual quantum circuit is hard unless the polynomial hierarchy collapses. In~\cite{movassagh2020quantum} the Cayley path was introduced, which is a unitary matrix-valued path. It was shown that the exact probability amplitudes of the (i.e.,~unitary) random quantum circuit is $\#\mathsf{P}$-hard, and that even approximating the amplitudes to within $2^{-m^c}$, where $c$ is a quantified constant and $m$ is the number of gates, remains $\#\mathsf{P}$-hard. The validity of the hardness with respect to additive error approximation is referred to as {\it robustness}.

\subsection{Summary of this work}
In this work we substantially~(super-polynomially) improve the robustness to $2^{-\Omega(m\log{m})}$, where $m$ is the number of gates. Therefore, our result proves that approximating the probability amplitudes to within $2^{-\Omega(n\log{n})}$ is hard for constant depth circuits.  In order to use Stockmeyer's algorithm to prove hardness of sampling from the hardness of approximating probability amplitudes with respect to additive errors, one needs a further property of anti-concentration. This property has been proved for circuits of depths $\log{n}$~\cite{dalzell2020random} and $\sqrt{n}$~\cite{harrow2018approximate}. Our robustness bound for circuits of depth $\log{n}$ and $\sqrt{n}$ is $2^{-\Omega\left(n\log^2n\right)}$ and $2^{-\Omega(n^{3/2}\log{n})}$ respectively. 

In proving this result, we took an entirely a new approach that does not follow the standard techniques of the past~\cite{aaronson2011computational,bouland2018quantum,movassagh2018efficient}.
Instead of using standard paths such as $\theta X+(1-\theta)Y$ in BosonSampling
~\cite{aaronson2011computational}, or the truncation of the Taylor series used in~\cite{bouland2018quantum}, we rely on the Cayley path introduced in~\cite{movassagh2020quantum}. Furthermore, instead of Paturi's lemma~\cite{paturi1992degree}, which has become standard in the field for bounding the polynomial extrapolation errors~\cite{aaronson2011computational,bouland2018quantum}, and Rakhmanov's result~\cite{rakhmanov2007bounds}, which is usually used to extend error bounds on a disrete set of points to a uniform error bound in a region~\cite{bouland2018quantum}, we use the Lagrange polynomials for estimating the error of polynomial extrapolation. In fact we do not use the well-known Berlekamp--Welch algorithm~\cite{welch1986error}, which has instabilities in the presence of uniform noise.

We prove two theorems for the hardness of RCS that complement one another by using oracles of different strengths yet requiring different success probabilities (Theorems \ref{thm:strong} and \ref{thm:weak}).
We prove our hardness results from first principles and hope that the new approach helps to overcome the insurmountable difficulties that the standard techniques meet. 
Our main results are the following two theorems:
\begin{thm}[Simplified]\label{thm:strong} It is $\sharpP$-hard under $\mathsf{BPP}$-reduction
to approximate $|\braket{0|C|0}|^2$ to within the additive error $2^{-\Omega(m \log m)}$ for $1-O(1/m)$ fraction of quantum circuits $C$.
\end{thm}
\begin{thm}[Simplified]\label{thm:weak} It is $\#\mathsf{P}$-hard under $\mathsf{BPP}^\mathsf{NP}$-reduction
to approximate $|\braket{0|C|0}|^2$ to within the additive error $2^{-\Omega(m \log m)}$ for $\frac34 + \frac1{\poly(n)}$ fraction of quantum circuits $C$.
\end{thm}

In other word, unless the polynomial hierarchy collapses to finite level, the above tasks are outside the polynomial hierarchy.
We summarize the proof structures that culminate in Theorems \ref{thm:strong} and \ref{thm:weak} in Figures \ref{fig:strong} and \ref{fig:weak} respectively.

\begin{figure*}[t]
    \centering
        \begin{tikzpicture}
        \node[rectangle,draw] (a) at (0, 6) {\begin{tabular}{c}
             1.~Random Circuit Sampling
        \end{tabular}};
        \node[rectangle,draw, label=right:$\therefore\; \sharpP$-hard under $\mathsf{BPP}$-reduction] (b) at (0, 3) {
                \begin{tabular}{l} 2.~Computing $|\braket{0|C|0}|^2$ on average\\with prob. $1-\frac{1}{\Omega(m)}$ to within\\
                additive error $2^{-\Omega(m\log m)}$
                \end{tabular}};
        \node[rectangle,draw,label=right:$\;\therefore\sharpP$-hard under $\mathsf{P}$-reduction] (c) at (0, 0) {
        \begin{tabular}{l} 3.~Computing $|\braket{0|C|0}|^2$ in the\\ worst case to within additive\\
                 error $2^{-\Omega(m\log m)}$
                \end{tabular}
            };
        \node[rectangle,draw] (d) at (0, -2.5) {4.Any  $\sharpP$ problem};
        \path[draw,->,dashed, thick] (b) -> (a) node[midway,right,text width=12cm] {
            {$\mathsf{BPP}^{\mathsf{NP}}$-reduction [Stockmeyer 1985]\\
            Available only when the additive error in Box 2 is $2^{-n}/\poly(n)$.}
            };
        \path[draw,->,thick] (c) -> (b) node[midway,right] {\begin{tabular}{c}
             $\mathsf{BPP}$-reduction:
             Lemma \ref{lem:strong_oracle}
        \end{tabular}};
        \path[draw,->,thick] (d) -> (c) node[midway,right] {$\mathsf{P}$-reduction: Lemma \ref{lem:SharpP}};
        \coordinate (p1) at (-4.5, 1);
        \path[draw,->, thick] (d.west) -| (p1) |- (b.west) node[left,midway] {Theorem \ref{thm:strong}};;
    \end{tikzpicture}
    \caption{The proof structure and reductions for Theorem \ref{thm:strong}}
    \label{fig:strong}
\end{figure*}
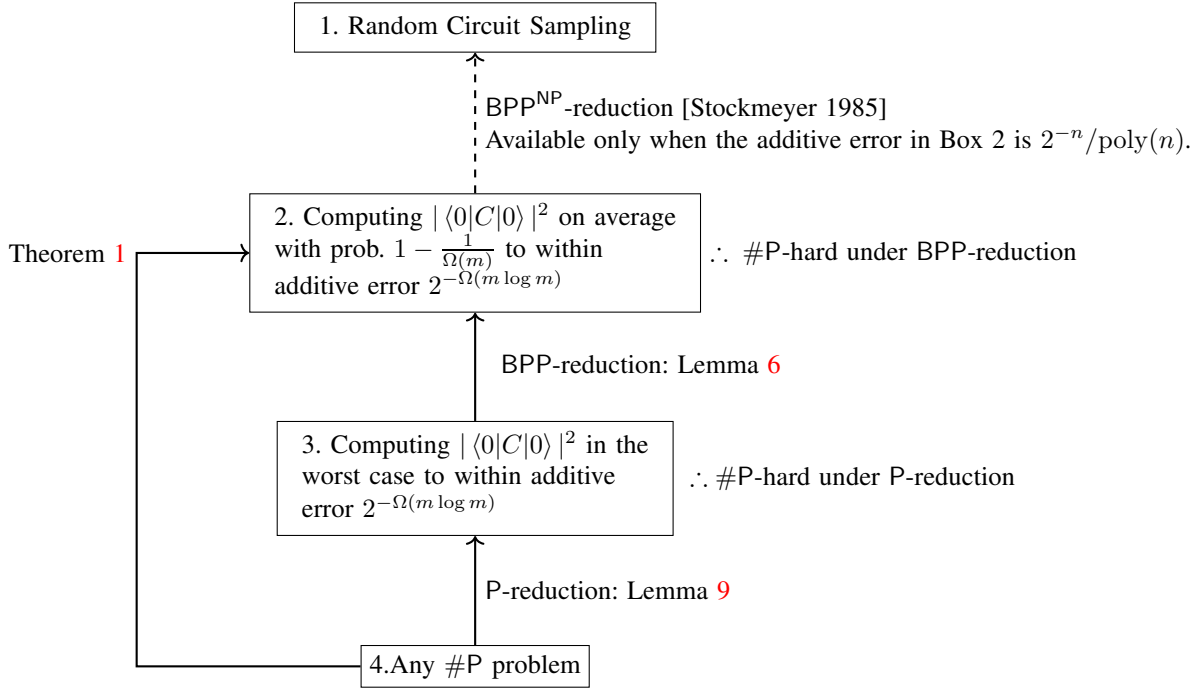
\begin{figure*}[t]
    \centering
        \begin{tikzpicture}
        \node[rectangle,draw] (a) at (0, 6) {\begin{tabular}{c}
             1.~Random Circuit Sampling
        \end{tabular}};
        \node[rectangle,draw, label=right:$\therefore\; \mathsf{\#P}$-hard under $\textsf{BPP}^\textsf{NP}$ reductions] (b) at (0, 3) {
                \begin{tabular}{l} 2.~Computing $|\braket{0|C|0}|^2$ on average\\ with prob. $3/4+1/\text{poly}(n)$ to within\\
                additive error $2^{-\Omega(m\log m)}$
                \end{tabular}};
        \node[rectangle,draw,label=right:$\therefore\;\mathsf{\#P}$-hard under $\mathsf{P}$-reduction] (c) at (0, 0) {
        \begin{tabular}{l} 3.~Computing $|\braket{0|C|0}|^2$ in the\\ worst case to within additive\\
                 error $2^{-\Omega(m\log m)}$
            \end{tabular}
            };
        \node[rectangle,draw] (d) at (0, -2.5) {4.Any $\mathsf{\#P}$ problem};
        \path[draw,->,dashed, thick] (b) -> (a) node[midway,right,text width=12cm] {
            {$\mathsf{BPP}^{\mathsf{NP}}$ reduction [Stockmeyer 1985]\\ Available only when the additive error in Box 2 is $2^{-n}/\poly(n)$.}
            };
        \path[draw,->,thick] (c) -> (b) node[midway,right] {\begin{tabular}{c}
             $\mathsf{BPP}^\mathsf{NP}$-reduction:
             Lemma \ref{lem:weak_oracle}
        \end{tabular}};
        \path[draw,->,thick] (d) -> (c) node[midway,right] {$\mathsf{P}$-reduction: Lemma \ref{lem:SharpP}};
        \coordinate (p1) at (-4.5, 1);
        \path[draw,->, thick] (d.west) -| (p1) |- (b.west) node[left,midway] {Theorem \ref{thm:weak}};;
    \end{tikzpicture}
    \caption{The proof structure and reductions for Theorem \ref{thm:weak}}
    \label{fig:weak}
\end{figure*}
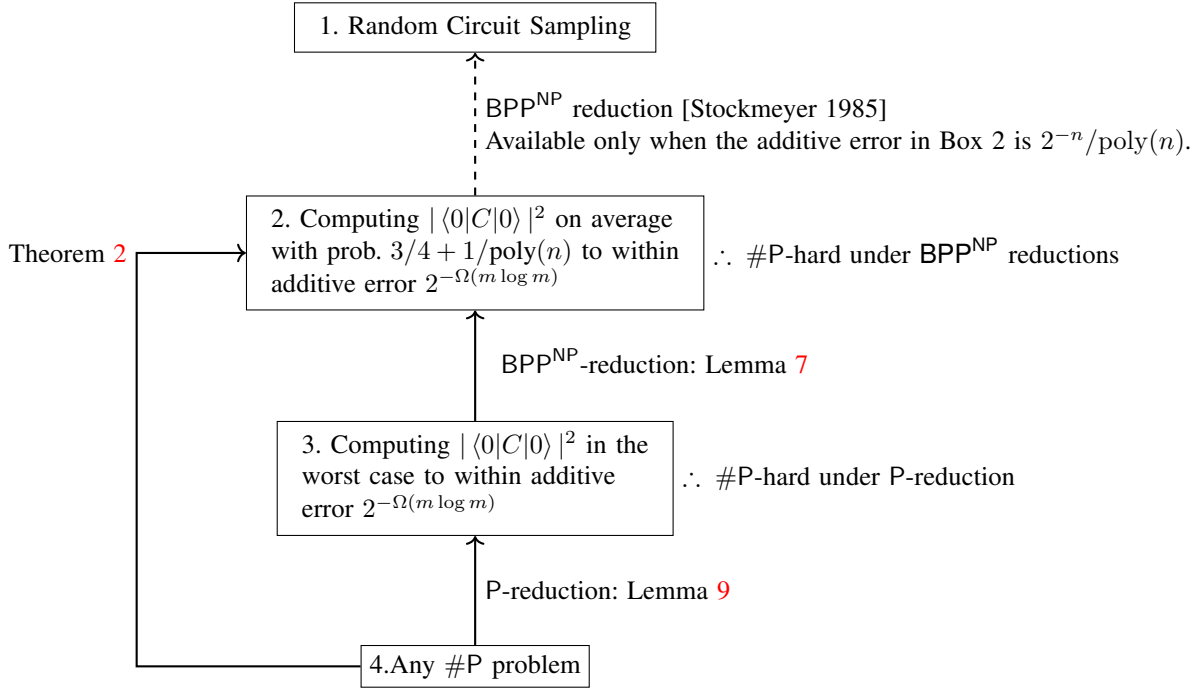

We then show that the Cayley path can be utilized in much the same way to prove that any circuit with the same architecture as the worst-case circuit also has the same hardness properties as shown above. A perhaps surprising corollary is that sampling from circuits close to identity as $\# P$-hard. 
\begin{thm}[Simplified]\label{thm:fixedU} Theorems \ref{thm:strong} and \ref{thm:weak} hold in the case that the circuit $C$ is within any open neighborhood of a fixed circuit. This also applies to the trivial circuit $C=I$ with identity gates (see Corollary \ref{cor:id}.)
\end{thm}

We remark that Bouland et al \cite{bouland2021noise} claim the same robustness as our Theorem~\ref{thm:weak} and arrived at it independently. Our respective papers are 
different in details. 
We use 
Lagrange interpolation for bounding the errors induced by polynomial extrapolation, which gives simple and direct proofs.
More complicated techniques were needed in~\cite{bouland2021noise} for bounding the errors due to polynomial extrapolation, which they call Lagrang{\it ian} interpolation.

\subsection{Open problems and future work}
Our results can be applied to BosonSampling. We hope to see extension and application of these techniques to the complexity of BosonSampling~\cite{KondoBoson2021} for which there was a recent experimental breakthrough \cite{zhong2020quantum}.

Applications of the Cayley path for randomizing quantum gates have found use in other contexts~\cite{oszmaniec2020fermion} and it would be interesting to see fresh new applications.

The main open problem is to prove the hardness of sampling for random quantum circuits. In particular, improving our additive error robustness bounds to $2^{-n}/\poly(n)$ for random circuits that have the anti-concentration property would be sufficient. 
The overarching goal of proving the quantum supremacy conjecture is achieved (i.e.,  Corollary~\ref{cor:QSupremacy} below is proved) if the following conjecture is proved in the affirmative:
\begin{conj}\label{conj:sampling_is_in_ph}
Approximating $|\braket{0|C|0}|^2$ with probability of $\frac{3}{4}+\frac{1}{\poly(n)}$ over the choice of quantum circuits $C$ to within the additive error $2^{-n}/\poly(n)$ implies the collapse of the polynomial hierarchy to a finite level.
\end{conj}
Assuming Conjecture~\ref{conj:sampling_is_in_ph} and anti-concentration of output probabilities, we obtain the hardness of RCS via Stockmeyer's theorem.
\begin{thm}[Stockmeyer~\cite{stockmeyer1985approximation}]\label{thm:stockmeyer}
Given a Boolean function $f:\{0,1\}^n\rightarrow \{0,1\}$, let 
\[
p=\Pr_{x\in\{0,1\}^n}[f(x)=1]=2^{-n}\sum_{x\in\{0,1\}^n}f(x).
\]
Then there exists an $\mathsf{FBPP}^{\mathsf{NP}^f}$ machine that approximates $p$ to within any multiplicative factor of $1+1/\poly(n)$.
\end{thm}

In conclusion, Theorem~\ref{thm:stockmeyer} along with the anti-concentration property of the output probabilities~\cite{harrow2018approximate}, and Conjecture~\ref{conj:sampling_is_in_ph} prove the following major open problem:
\begin{conj}\label{cor:QSupremacy}
Classically sampling from any distribution with a total variation distance of $1/\poly(n)$ from the output distribution of the random quantum circuit is hard unless the polynomial hierarchy collapses to finite level.
\end{conj}

\section{Cayley path}

We now define the interpolation between any two gates of the quantum computation (i.e., any two unitaries) based on the Cayley path, which was first introduced in \cite{movassagh2020quantum}.
Suppose $U_{0},U_{1}\in\mathbb{U}(N)$ are unitary matrices and we wish to interpolate between $U_0$ and $U_1$ via a path with nice algebraic properties that can be utilized in our reductions below. Let $\theta\in\mathbb{R}$
and $f(\theta)$ be the \textit{Cayley function}
\begin{equation}
f(\theta)=\frac{1+i\theta}{1-i\theta}\, ,\label{eq:f_x}
\end{equation}
where one defines $f(-\infty)=-1$. The Cayley function as just defined is a bijection between $\mathbb{R}\cup\{-\infty\}$ and the unit circle in the complex plane. 

The proposed path is
\begin{equation}
U(\theta)=U_{0}\:f(\theta h)=\sum_{\alpha=1}^{N}f(\theta h_{\alpha})\;U_{0}|\psi_{\alpha}\rangle\langle\psi_{\alpha}|\, ,\label{eq:C_THETA}
\end{equation}
where $h$ is a hermitian matrix defined by $h=f^{-1}\left(U_0^\dagger U_1\right)$, which is guaranteed to exist by the bijection property.
$U(\theta)$ is a unitary matrix as it is a product of two unitary
matrices. Note that $U(0)=U_{0}f(0)=U_{0}$ and $U(1)=U_{0}U_{0}^{\dagger}U_{1}=U_{1}$
as desired. We now derive the algebraic dependence of the entries
of $U(\theta)$ on $\theta$.

Let the spectral decomposition of $h$ be $h=\sum_{\alpha=1}^N h_\alpha |{\psi_\alpha}\rangle\langle\psi_\alpha|$.
Using the definition of the Cayley function and the foregoing equations
we write
\begin{equation}
U(\theta)=\frac{1}{q(\theta)}\sum_{\alpha=1}^{N}p_{\alpha}(\theta)\text{ }U_{0}\text{ }|\psi_{\alpha}\rangle\langle\psi_{\alpha}|\:,\label{eq:Cayley_path}
\end{equation}
where $q(\theta)$ and $p_{\alpha}(\theta)$ are univariate scalar complex polynomials
of degree $N$ in $\theta$:
\begin{align}
q(\theta)&=\prod_{\alpha=1}^{N}(1-i\theta h_{\alpha})\nonumber\\ p_{\alpha}(\theta)&=(1+i\theta h_{\alpha})\prod_{\beta\in[N]\backslash\alpha}(1-i\theta h_{\beta}).\label{eq:q_x}
\end{align}
In what follows we wish to apply the Cayley path to each one of the gates in the quantum circuit. For this we first need to formally define the architecture of the circuit.
\begin{dfn}[Architecture]\label{def:architecture}
The architecture $\mathcal{A}$ is a poly-time uniform family $\{A_\ell\}_{\ell=1,2,\dotsc}$ of quantum circuit where all quantum gates are ``blank'' and not specified. The quantum circuit $A_\ell$ has $n(\ell)$ qubits and $m(\ell)$ quantum gates.
If all quantum gates of $A_\ell$ are specified, then the quantum circuit is specified. We say an architecture is {\it local} if each gate acts on at most two qubits.
\end{dfn}
In this paper all architectures are assumed to be local.
\begin{dfn}
Let $\mathcal{H_{A}}$ be the distribution over circuits with architecture
$\mathcal{A}$ with an implicitly fixed $\ell$ (see Definition~\ref{def:architecture}), whose local gates are unitary matrices drawn independently and at random
from the Haar measure. 
\end{dfn}
Given a fixed architecture and a quantum circuit $C$ whose $k$-th gate is $C_k$, we consider the randomized quantum circuit by replacing all quantum gates of $C$ with quantum gates drawn from the Haar measure.  Then, we consider a Cayley path interpolation between each gate of the fixed circuit and the randomized one.

Let $C(\theta)$ denote the quantum circuit with $m$ gates:
\begin{equation}\label{eq:C}
  C(\theta)=\mathcal{C}_m(\theta)\cdots \mathcal{C}_2(\theta)\,\mathcal{C}_1(\theta)\, ,  
\end{equation}
where $\mathcal{C}_k(\theta)=\mathbb{I}\otimes C_k(\theta)$ is a unitary matrix that only acts non-trivially on the qubits that $C_k(\theta)$ acts on. Here each local unitary gate is a unitary-valued Cayley path $C_k(\theta)$:
\begin{equation}
C_{k}(\theta)=C_{k}\:f(\theta h_{k})\;,\label{eq:CayleyPathSM}
\end{equation}
where $f(\theta h_{k})$ is a unitary matrix and $h_k$ is hermitian  $h_{k}^{\dagger}=h_{k}$. Suppose $C_k$ is a fixed gate of a quantum computation and $H_k\equiv f(h_k)$ is a Haar unitary matrix then $C_{k}(0)=C_{k}$. Moreover, by the translation invariance of the Haar measure $C_{k}(1)=C_{k}H_k$ is a Haar random gate. Hence we have an interpolation scheme between any fixed gate and a Haar random gate. 
\begin{dfn}\label{def:H_A}
Let us denote by $\mathcal{H}_{\mathcal{A},\Delta}$
the distribution over circuits whose local gates are drawn from the distribution
induced by the Cayley path for $\theta=1-\Delta$ for $\Delta\in[0,1]$.
\end{dfn}
The randomness of the quantum circuit under the Cayley path is quantified in this lemma:
\begin{lem}[Total Variation Distance~\cite{movassagh2020quantum}]\label{lem:tvd}
For a circuit with $m$ gates and an architecture $\mathcal{A}$, the total variation distance between $\mathcal{H}_{\mathcal{A}}$ and $\mathcal{H}_{\mathcal{A},\Delta}$  is $O(m\Delta)$. 
\end{lem}
We make the dependence on $k$ explicit in Eqs.~\eqref{eq:Cayley_path} and \eqref{eq:q_x} by denoting $p_{\alpha}(\theta)\mapsto p_{k,\alpha}(\theta)$ and $q(\theta)\mapsto q_{k}(\theta)$.
We can now express Eq.~\eqref{eq:CayleyPathSM} as
\begin{equation}
C_{k}(\theta)=\frac1{q_{k}(\theta)}\sum_{\alpha=1}^{N}p_{k,\alpha}(\theta)\;C_{k}|\psi_{k,\alpha}\rangle\langle\psi_{k,\alpha}|\label{eq:C_theta}
\end{equation}
where $q_{k}(\theta)  =  \prod_{\alpha=1}^{N}(1-i\theta h_{k,\alpha})$
\begin{equation}
p_{k,\alpha}(\theta)  =  \sum_{\alpha=1}^{N}(1+i\theta h_{k,\alpha})\prod_{\beta\in[N]\backslash\alpha}(1-i\theta h_{k,\beta}).\label{eq:p_theta}
\end{equation}
It will be useful to make a change of variables to $\theta=1-x$ such
that generic instances correspond to $x=0$ and $\mathsf{\#P}$-hard point
to $x=+1$. 

The probability amplitude of starting the quantum computation in the state $|0^n\rangle$ and measuring the string $|0^n\rangle$ is $p_0(x)\equiv|\langle0^{n}|C(x)|0^{n}\rangle|^{2}$. Note that at $x=1$ we recover the worst case $\mathsf{\#P}$-hard instance probability amplitude and $x=0$
corresponds to the probability amplitude of the generic random circuit. Using $\theta=1-x$ in Eqs.~\eqref{eq:C_theta}-\eqref{eq:p_theta}, the circuit has
the algebraic form~\cite{movassagh2020quantum}
\begin{align}
|\langle0^{n}|C(x)|0^{n}\rangle|^{2} &= \left|\langle0^{n}|\prod_{k=1}^{m}\mathcal{C}_{k}(x)|0^{n}\rangle\right|^{2}\nonumber\\
&\equiv\frac{|\langle0^{n}|P(x)|0^{n}\rangle|^{2}}{|Q(x)|^{2}},\label{eq:P_over_Q-1-1}
\end{align}
where
\begin{align}
P(x) \equiv \sum_{\alpha_{1},\dots,\alpha_{m}=1}^{N}\prod_{k=1}^{m}g_{k,\alpha_{k}}(x)\:\mathcal{C}_k\left(x=1\right)|\psi_{k,\alpha_{k}}\rangle\langle\psi_{k,\alpha_{k}}|\label{eq:rationality_C-1-1}
\end{align}
and
\begin{align}
|Q(x)|^{2} & \equiv \prod_{k=1}^{m}\prod_{\alpha_{k}=1}^{N}\left|1+ix\frac{h_{k,\alpha_{k}}}{r_{k,\alpha_{k}}}e^{iu_{k,\alpha_{k}}}\right|^{2}\label{eq:Q_Theta-1-1}\\
g_{k,\alpha_{k}}(x) & \equiv \left[e^{iu_{k,\alpha_{k}}}-ix\frac{h_{k,\alpha_{k}}}{r_{k,\alpha_{k}}}\right]\nonumber\\
&\times\prod_{\beta_{k}\in[N]\backslash\alpha_{k}}\left[e^{-iu_{k,\beta_{k}}}+ix\frac{h_{k,\beta_{k}}}{r_{k,\beta_{k}}}\right].\nonumber 
\end{align}
Here, we let $1\pm ih_{k,\alpha_k}=r_{k,\alpha_k}e^{\pm iu_{k,\alpha_k}}$ with $r_{k,\alpha_k}$ and $u_{k,\alpha_k}$ defined as $r_{k,\alpha_k}=\sqrt{1+h^2_{k,\alpha_k}}$ and $u_{k,\alpha_k}=\arctan(h_{k,\alpha_k})$.

The quantity $|Q(x)|^{2}$ can be pre-computed
in time $\Theta(m)$ as it only depends on the eigenvalues of the local terms which are matrices of size at most $N=4$. Since $\frac{h_{k,\alpha_{k}}}{r_{k,\alpha_{k}}}<1$
and for generic circuits $|x|\le\Delta=O(m^{-1})$, it is easily seen that $|Q(x)|^2$ is very near one:
\begin{align}
|Q(x)|^{2}&\le\prod_{k=1}^{m}\prod_{\alpha_{k}=1}^{N}\left|1+ix\frac{h_{k,\alpha_{k}}}{r_{k,\alpha_{k}}}e^{iu_{k,\alpha_{k}}}\right|^{2}\nonumber\\
&\le1+O(m\Delta)\;.\label{eq:Qx_bound}
\end{align}

\section{Proof of average-case robustness}
Our goal here is to prove that \textit{approximating} $p_{0}(x)\equiv|\langle0^{n}|C(x)|0^{n}\rangle|^{2}$ to within $\epsilon$ additive error is \textit{hard}
for as large an $\epsilon$ as possible. That is given an $x_{i}$ and a classical algorithm that promises to give us $p_{0}(x_{i})+\epsilon_{i}$
efficiently (polynomial classical time), where $|\epsilon_{i}|\le\epsilon\ll1$,
we wish to construct a low degree algebraic function $\tilde{p}(x)$ whose
extrapolation to $x=1$ is guaranteed to be \textit{hard}.

Since for any $x$, $|Q(x)|^{2}$ can be computed in time $\Theta(m)$ we can reduce the rational
functional form of $p_{0}(x)$ to a polynomial by multiplying through
by $|Q(x)|^{2}$, which for any given $x$ can be treated as simply
a constant. Let us denote by the ``exact'' polynomial
\begin{equation}\label{eq:Px}
   p_{e}(x)=|\langle0^{n}|P(x)|0^{n}\rangle|^{2}=|Q(x)|^{2}p_{0}(x)
\end{equation}
of degree $8m$ where we treated $Q(x)$ as a known constant. Therefore, we have at
our disposal a set of tuples $(x_{i},p_{e}(x_{i})+\epsilon_{i}|Q(x_{i})|^{2})$. In Eq.~\eqref{eq:Qx_bound} we showed that $|Q(x)|^{2}\le1+O(m\Delta)$, and
by taking $\Delta=O(m^{-1})$ we are guaranteed to have $|Q(x)|^{2}\approx1$.
This shows that the additive error $|Q(x_{i})|^{2}\epsilon_{i}\approx\epsilon_{i}$. 

Let the difference of the exact polynomial $p_{e}(x)$ from the one
that results from the extrapolation of the erroneous polynomial
$\tilde{p}(x)$ be defined by
\[
p(x)\equiv \tilde{p}(x)-p_{e}(x).
\]
We are promised that for all $|x_{i}|\le\Delta$, $|p(x_{i})|\le\epsilon$
and wish to show that $|p(1)|$ is sufficiently small such that it falls within a region whose hardness is guaranteed. We will return to the quantification of this region in Section~\ref{sec:CEQP}. For now let us bound the polynomial extrapolation error $|p(1)|$.  

The following Lemmas \ref{lem:PaturiLemm} and \ref{lem:3} and Corollary \ref{cor:PaturiCor} on this page are presented to show that one may obtain robustness beyond what the traditional Paturi's lemma allows. Nevertheless this approach will face other difficulties. Hence, starting in Lemma \ref{lem:Robustness} we prove the results  independent of this approach.

Traditionally this bound is obtained using Paturi's lemma \cite{paturi1992degree},
which we recall:
\begin{lem}\label{lem:PaturiLemm}[Paturi's lemma \cite{paturi1992degree}] Let $p(x)$ be a
polynomial of degree $d$, and suppose $|p(x)|\le\epsilon$ for $|x|\le\Delta$ where $\Delta\in(0,1)$.
Then $p(1)\le\epsilon\exp[2d(1+\Delta^{-1})]$
\end{lem}
For $k\ge0$, let us denote by $T_{k}(x)$ the $k^{\text{th}}$ Chebyshev
polynomial, which is a degree $k$ algebraic polynomial defined by
\[
T_{k}(x)=\frac{1}{2}\left[(x+\sqrt{x^{2}-1})^{k}+(x-\sqrt{x^{2}-1})^{k}\right]
\]
for $x > 1$.
Paturi has another result in the same paper (Corollary 2 in \cite{paturi1992degree}),
which says
\begin{cor}\label{cor:PaturiCor}
[Paturi's Corollary \cite{paturi1992degree}] Let $p(x)$ be a polynomial
of degree at most $d$. Assume $|p(x)|\le\epsilon$ in the interval
$[-\Delta,\Delta]$ for some $0<\Delta\le1$. We then have $|p(x)|\le\epsilon\left|T_{d}(1+\frac{|x|-\Delta}{\Delta})\right|$
for all $|x|\ge\Delta$ where $T_d$ denotes the Chebyshev polynomial of degree $d$.
\end{cor}
We shall use the latter and prove the following lemma
\begin{lem}\label{lem:3}
Let $p(x)$ be a polynomial of degree $d$, and suppose $|p(x)|\le\epsilon$
for $|x|\le\Delta$ where $\Delta\in(0,1)$. Then
\[
|p(1)|<\epsilon\,\exp(d\log|2\Delta^{-1}|).
\]
\end{lem}
\begin{proof}
From Paturi's corollary we have $|p(1)|\le\epsilon|T_{d}(\Delta^{-1})|$.
Moreover,
\begin{align*}
|T_{d}(x)| & \le \frac{1}{2}\left[\left|(x+\sqrt{x^{2}-1})^{d}\right|+\left|(x-\sqrt{x^{2}-1})^{d}\right|\right]\\
 & < \left(|x|+\sqrt{x^{2}-1}\right)^{d}<|2x|^{d}=e^{d\log(2|x|)}.
\end{align*}
We conclude that $|p(1)|<\epsilon\,\exp\left[d\log(2\Delta^{-1})\right]=\epsilon\,\exp\left[-\left(d\log\Delta\right)\left(1-\frac{1}{\log_{2}\Delta}\right)\right]$.
\end{proof}
An issue one faces in making robustness claims is that the discrete
bound of $|p(x_{i})|\le\epsilon$ for $|x_{i}|\le\epsilon$, does
not readily imply a uniform bound $|p(x)|\le\epsilon$ for all $|x|\le\Delta$.
This is traditionally remedied by Rakhmanov's result \cite{rakhmanov2007bounds} (see \cite{movassagh2020quantum,bouland2018quantum}).
Here we will do without Rakhmanov's result. To do so, take $d+1$
points in the interval $[-\Delta,\Delta]$ and estimate
the function $p(x)$ using the Lagrange interpolation technique. We
prove
\begin{lem}
\label{lem:Robustness}Let $p(x)$ be a polynomial of degree at most
$d$. Let $\Delta\in(0,1)$. Assume that $|p(x_{j})|\le\epsilon$ for all of the $d+1$ equally-spaced points
$x_{j}=-\Delta+\frac{2j}{d}\Delta$ for $j=0,1,\dots,d$. Then
\begin{equation}
|p(1)|<\epsilon\:\frac{\exp\left[d(1+\log\Delta^{-1})\right]}{\sqrt{2\pi d}}\;.\label{eq:Error_Final}
\end{equation}
\end{lem}
\begin{proof}
Let $p_{j}=p(x_{j})$ for all $j=\{0,1,2,\dots,d\}$, where by assumption
$|p_{j}|\le\epsilon$. The Lagrange representation of the function $p(x)$
writes
\[
p(x)=\sum_{j=0}^{d}p_{j}\;\delta_{j}(x),\qquad\delta_{j}(x)\equiv\frac{\prod_{\ell\ne j}x-x_{\ell}}{\prod_{\ell\ne j}x_{j}-x_{\ell}}\;.
\]
By triangular inequality we have $|p(1)|\le\epsilon\sum_{j=0}^{d}|\delta_{j}(1)|$.
Moreover, using $x_{j}=-x_{d-j}$ and the fact that $|x_{j}|<1$ for
all $j$ we have
\begin{align*}
&|\delta_{j}(1)|=\frac{\prod_{\ell\ne j}|1-x_{\ell}|}{\prod_{\ell\ne j}|x_{j}-x_{\ell}|}\\
&=\frac{(1+x_{j})\prod_{\ell\in\{0,1,\dotsc,\lfloor(d-1)/2\rfloor\}\setminus\{j,\,d-j\}}(1-x_{\ell}^{2})}{\prod_{\ell\ne j}|x_{j}-x_{\ell}|}\\
&<\frac{(1+x_{j})}{\prod_{\ell\ne j}|x_{j}-x_{\ell}|}.
\end{align*}
Since $x_{j}-x_{\ell}=\frac{2\Delta}{d}(j-\ell)$, we have $\prod_{\ell\ne j}|x_{j}-x_{\ell}|=(\frac{2\Delta}{d})^{d}\prod_{\ell\ne j}|j-\ell|$.
Moreover $\prod_{\ell\ne j}|j-\ell|=\prod_{\ell\in\{0,1,2,\dots,,j-1,j+1,\dots,d\}}|j-\ell|=j!(d-j)!$
and we obtain
\[
|\delta_{j}(1)|=\left(\frac{d}{2\Delta}\right)^{d}\frac{(1+x_{j})}{\prod_{\ell\ne j}|j-\ell|}=\left(\frac{d}{2\Delta}\right)^{d}\frac{(1+x_{j})}{j!\:(d-j)!}\;.
\]
We express the bound  $|p(1)|\le\epsilon\sum_{j=0}^{d}|\delta_{j}(1)|$
as
\begin{align*}
|p(1)| & < \epsilon\left(\frac{d}{2\Delta}\right)^{d}\sum_{j=0}^{d}\frac{(1+x_{j})}{j!\:(d-j)!}\\
&=\epsilon\left(\frac{d}{2\Delta}\right)^{d}\frac{1}{d!}\sum_{j=0}^{d}\binom{d}{j}(1+x_{j}).
\end{align*}
Using the symmetry of $x_{j}=-x_{d-j}$ we have $(1+x_{j})+(1+x_{d-j})=2$
and it is easy to see that (irrespective of the parity of $d$)
\[
\sum_{j=0}^{d}\binom{d}{j}(1+x_{j})=2^{d}.
\]
By Stirling's inequality $n! \ge \sqrt{2\pi n}\frac{n^n}{\mathrm{e}^n}$, we conclude that
\begin{align*}
|p(1)|&<\epsilon\left(\frac{d}{2\Delta}\right)^{d}\frac{2^{d}}{d!}\le\epsilon\frac{(\mathrm{e}\Delta^{-1})^{d}}{\sqrt{2\pi d}}\\
&=\epsilon\frac{\exp\left[d(1+\log\Delta^{-1})\right]}{\sqrt{2\pi d}}\;.
\end{align*}
\end{proof}
\begin{rem}
Note that the choice of the equally-spaced $d+1$ points is not exactly optimal. We choose the $d+1$ points $\{\Delta\cos(0),\Delta\cos(\pi/d),\Delta\cos(2\pi/d),\dotsc,\Delta\cos((d-1)\pi/d),\Delta\cos(\pi)\}$, which are the extrema of the Chebyshev polynomials $T_d(x/\Delta)$. Then
the worst-case polynomial is the Chebyshev polynomial $\epsilon T_d(x/\Delta)$, and the same bound as in Lemma~\ref{lem:3} is obtained. The issue with Lemma~\ref{lem:3} remains to be that it either assumes a uniform bound or requires using Chebyshev extrema as just described. To prove Theorem~\ref{thm:weak} we need Lemma~\ref{lem:Linterpolate}. This lemma does not give the freedom to choose the points exactly. Therefore, we cannot use Chebyshev extrema at will that saturate the Paturi's Lemma (Lemma~\ref{lem:3})
\end{rem}

Similarly, we obtain the following lemma where $d+1$ points are chosen from $L$ equally-spaced points in $[-\Delta,\Delta]$.
\begin{lem}\label{lem:Linterpolate}
Let $p(x)$ be a polynomial of degree at most $d$, and $L$ an integer at least $d+1$.
Let $a_0,a_1,\dotsc,a_d$ be integers satisfying $0\le a_0 < a_1 < \dotsb < a_d\le L-1$.  Let $\Delta\in(0,1)$.
Assume that $|p(x_{j})|\le\epsilon$ for all of the $d+1$ points
$x_{j}=-\Delta+\frac{2a_j}{L-1}\Delta$ for $j=0,1,2,\dots,d$. Then
\begin{align*}
    |p(1)|
    \leq \epsilon \, \frac{\exp [d (1 + \log ((1+\Delta^{-1})\frac{L-1}{d}))]}{\sqrt{2\pi d}}\;.
\end{align*}
\end{lem}
The proof of Lemma~\ref{lem:Linterpolate} is similar to the proof of Lemma~\ref{lem:Robustness}, and is presented in Appendix~\ref{apx:Linterpolate}.

Armed with the new extrapolation error bound, we proceed to prove
the hardness of evaluating the probability amplitudes of general circuits
with an additive error. 
In the following, we show the reductions from the worst-case to the average-case computation of the output probability of the quantum circuit. While real numbers appear in the reduction algorithms, they should be represented by $\mathrm{poly}(m)$ bits. The rounding only causes additional errors of size $2^{-\mathrm{poly}(m)}$, which will not affect our bounds and conclusions. As in~\cite{aaronson2011computational} and for simplicity, we ignore rounding issues in the proofs,
and work with real numbers.
\begin{lem}[Strong oracle] \label{lem:strong_oracle}
 Let $\delta>0$ be a constant and $\mathcal{O}$ a classical oracle that takes as input the classical description of the quantum circuit $C$ in the architecture $\mathcal{A}$ and outputs $\mathcal{O}(C)$ that satisfies
\[
\Pr_{C\sim\mathcal{H_{A}}}\left[\;|\,\mathcal{O}(C)-|\langle0^{n}|C|0^{n}\rangle|^{2}\,|\le\epsilon\;\right]\ge1-\frac{\delta}{8m+1}.
\]
Then there exists a classical probabilistic polynomial-time algorithm $\mathcal{R}$ with access to $\mathcal{O}$ that outputs $\mathcal{R}^{\mathcal{O}}(C)$ satisfying:
\begin{align*}
&\Pr_{\mathcal{R}}\left[\;|\,\mathcal{R}^{\mathcal{O}}(C)-|\langle0^{n}|C|0^{n}\rangle|^{2}\,|<\epsilon\exp\left[O(m\log m)\right]\;\right]\\
&\ge1-\delta-\frac1{\poly(m)}.
\end{align*}
\end{lem}
\begin{proof}
Previously we proved that the total variation distance between $\mathcal{H_{A}}$
and $\mathcal{H}_{\mathcal{A},\Delta}$ is $O(m\Delta)$ (see Lemma~\ref{lem:tvd}).
 Hence invoking the oracle $\mathcal{O}$ it holds that
\begin{align*}
&\Pr_{C\sim\mathcal{H}_{\mathcal{A},\Delta}}\left[\;|\,\mathcal{O}(C)-|\langle0^{n}|C|0^{n}\rangle|^{2}\,|>\epsilon\;\right]\\
&\le\frac{\delta}{8m+1}+O(m\Delta).
\end{align*}
By the union bound, the probability that at least one of the $8m+1$
points evaluated by $\mathcal{O}$ has error larger than $\epsilon$
is at most $\delta+(8m+1)O(m\Delta)$. By choosing $\Delta=\Theta(m^{-k})$ for some constant $k> 2$, from Lemma~\ref{lem:tvd}, the error probability is at most $\delta + 1/\poly(m)$.
From Lemma \ref{lem:Robustness} we
know that if all $8m+1$ evaluation points have an error at most $\epsilon$,
then the extrapolation error from $x\in[-\Delta,\Delta]$ to $x=1$ in the
Lagrange extrapolation is given (via Eqs.~\eqref{eq:Qx_bound} and~\eqref{eq:Error_Final})
\begin{align*}
&\epsilon\:|Q(1+\Delta)|^2\:\frac{\exp\left[8m(1+\log\Delta^{-1})\right]}{\sqrt{16\pi m}}\\
&\le\epsilon\:(1+O(m\Delta))\frac{\exp\left[8m(1+\log\Delta^{-1})\right]}{\sqrt{16\pi m}}.
\end{align*}
Since $\Delta=\Theta(m^{-k})$ for some constant $k> 2$, we obtain the Lemma.
\end{proof}
It is desirable to make the oracle $\mathcal{O}$ as weak as possible.
In the following Lemma we show that this can be done at the expense
of introducing an \NP-machine. 

\begin{algorithm*}[t]
\begin{algorithmic}
\Function{$\mathcal{R}$}{$C$}
\State Draw a fixed quantum circuit $H$ according to $\mathcal{H}_{\mathcal{A}}$
\For{$i\in \{0,\dotsc,L-1\}$}
\State $x_i\gets \Delta (2i-L+1)/(L-1)$
\State $y_i\gets
\mathcal{O}\left(C(x_i)\right)\, |Q(x_i)|^2$
\Comment{\parbox[t]{6.5cm}{$C(x_i)\in\mathcal{H}_{\mathcal{A},\Delta}$ with $H$ kept fixed for all~$x_i$. 
See Eqs.~\eqref{eq:C}--\eqref{eq:P_over_Q-1-1} and Def.~\ref{def:H_A}.}}
\EndFor
\State $l\gets 0$
\State $r\gets 2$
\Loop \quad $\mathrm{poly}(m)$ times
  \State $c\gets (l + r) / 2$
  \If{$W\left(1_d, (x_i,y_i)_{i=0}^{L-1}, l, c\right)$} \Comment{$W$ is an \NP-oracle}
    \State $r\gets c$
  \Else
    \State $l\gets c$
  \EndIf
\EndLoop
\State \Return $l$ \Comment{$l$ is the approximation for $\tilde{p}(1)$}
\EndFunction
 \end{algorithmic}
\caption{The reduction algorithm $\mathcal{R}$ where $d\defeq 8m$ and $L\defeq \lceil(d+1)/\delta\rceil$}\label{alg:R}
\end{algorithm*}

The high level idea for the following lemma is that if we know the points $x_i$ for which the additive error committed by the classical oracle $O$ is sufficiently small then we can apply Lemma~\ref{lem:Linterpolate} to do Lagrange extrapolation. However, to find those points we need to call an \NP-oracle. Then we call the classical oracle $O$ that succeeds in approximating $|\langle0^{n}|C|0^{n}\rangle|^{2}$ with probability at least $\frac{3}{4}+\delta$ over the random choice of quantum circuits.

\begin{lem}[Weak oracle]\label{lem:weak_oracle}
Let $\delta>0$ be a constant and $\mathcal{O}$ a classical oracle that takes as input the classical description of the quantum circuit $C$ in the architecture $\mathcal{A}$ and outputs $\mathcal{O}(C)$ that satisfies
\[
\Pr_{C\sim\mathcal{H_{A}}}\left[\;|\,\mathcal{O}(C)-|\langle0^{n}|C|0^{n}\rangle|^{2}\,|\le\epsilon\;\right]\ge\frac{3}{4}+\delta.
\]
Then there exists a classical probabilistic algorithm $\mathcal{R}$ with the oracle access to $\mathcal{O}$ and an \NP-machine which outputs $\mathcal{R}^{\mathcal{O,\NP}}(C)$ in time $\poly(n, m)$ satisfying:
\begin{align}
&\Pr_{\mathcal{R}}\Bigl[\big|\mathcal{R}^{\mathcal{O},\NP}(C)-|\langle0^{n}|C|0^{n}\rangle|^{2}\big|<\epsilon
\exp\left[O(m\log \frac{m}{\delta})\right]\Bigr]\nonumber\\
&\ge\frac{1}{2}+\delta-\frac1{\poly(m)}\;.\label{eq:WeakLemma}
\end{align}
\end{lem}
\begin{proof}
We first describe the probabilistic algorithm $\mathcal{R}$ for the reduction.
Then, we show that the algorithm $\mathcal{R}$ satisfies the conditions in the lemma.
We define the \NP-oracle $W$ as the oracle that solves the following \NP problem:

\begin{itemize}[leftmargin=*, align=left]
\item[{\bf (Input)}] A positive integer $d$ in the unary representation, $L$ pairs $\{(x_{i},y_{i})\in\mathbb{R}^2\}_{i\in\{0,1,\dotsc,L-1\}}$,
and $l, r\in\mathbb{R}$ such that $l<r$.

\item[{\bf (Output)}] {\it True}: if there exists a polynomial $\tilde{p}(x)=\sum_{j=0}^{d}a_{j}x^{j}$
such that
\begin{align*}
    &\left|\{i\in\{0,\dotsc,L-1\}\, :\, |\tilde{p}(x_{i})-y_{i}|\le|Q(x_i)|^2\epsilon\}\right|\\
    &\ge (1+\delta)L/2,
\end{align*}
and $\tilde{p}(1)\in[l,r)$. {\it False}: otherwise.
\end{itemize}

This problem is in \NP since for a given certificate $\tilde{p}(x)$,
the conditions above can be verified in polynomial time.
We now describe the reduction algorithm $\mathcal{R}$ shown in Algorithm~\ref{alg:R}. Then, we will show that algorithm $\mathcal{R}$ satisfies the conditions of the lemma.

$\mathcal{R}$ is a probabilistic polynomial-time algorithm accessing the oracle $\mathcal{O}$ and the \NP-oracle $W$.
 We set $\Delta=\Theta(m^{-k})$ for some constant $k\ge 2$ 
so that the total variation distance between $\mathcal{H}_\mathcal{A}$ and $\mathcal{H}_{\mathcal{A},\Delta}$ is  $O(m\Delta)=1/\poly(m)$ 
by Lemma~\ref{lem:tvd}.
By $L$ calls to the oracle $\mathcal{O}$, the probability that at least $(1+\delta)L/2$ points are computed
with error of at most $\epsilon$ is at least $1/2+\delta-1/\poly(m)$; the latter follows from Markov's inequality (see
Appendix~\ref{apx:conc} for details).
So we can assume that at least $(1+\delta)L/2$ points are computed with error at most $\epsilon$.
In other words,
\begin{align*}
    &\left|\{i\in\{0,\dotsc,L-1\}\, : \, |p_e(x_{i})-y_{i}|\le |Q(x_i)|^2\epsilon\}\right|\\
    &\ge (1+\delta)L/2
\end{align*}
for the degree-$d$ polynomial $p_e(x) \defeq  |\bra{0^n} P(x)\ket{0^n}|^2$ in Eq.~\eqref{eq:Px}.
In this case, $W\left(1_{d}, (x_i,y_i)_{i=0}^{L-1}, 0,2\right)$ is true where $[0,2)$ is the initial region for the binary search since $p_e(x)$ can be seen as a certificate that satisfies the conditions in $W$, and by Eq.~\eqref{eq:Qx_bound} we have $p_e(x)\le |Q(x)|^2< 2$.
From the binary search in Algorithm~\ref{alg:R}, we obtain $l$ and $r$ such that $0\le r-l\le 2^{-\poly(m)}$ and $W(1_d,(x_i,y_i)_{i=0}^{L-1},l,r)$ is true. Let $\widetilde{p}(x)=\sum_{j=0}^da_jx^j$ be a certificate for $W(1_d,(x_i,y_i)_{i=0}^{L-1},l,r)$. Then, $\widetilde{p}(x)$ satisfies
\begin{align*}
 &\left|\{i\in\{0,\dotsc,L-1\}\, : \, |\tilde{p}(x_{i})-y_{i}|\le|Q(x_i)|^2\epsilon\}\right|\\
 &\ge (1+\delta)L/2\;,
\end{align*}
 and $|\tilde{p}(1)-\mathcal{R}(C)|\le2^{-\mathrm{poly}(m)}$ because $\mathcal{R}(C)=l$ and $\widetilde{p}(1)\in[l,r)$.
 
 In the following, we show that $|\tilde{p}(1) - p_e(1)| \le \epsilon\exp\{O(m\log m)\}$.
Define the two sets $S_{p_e}$ and $S_{\tilde{p}}$ by
\begin{align*}
S_{p_e} & =  \left\{ i\in\{0,\dotsc,L-1\}\; : \;\left|p_e(x_{i})-y_{i}\right|\le|Q(x_i)|^2\epsilon \right\} ,\\
S_{\tilde{p}} & =  \left\{ i\in\{0,\dotsc,L-1\}\; : \;\left|\tilde{p}(x_{i})-y_{i}\right|\le|Q(x_i)|^2\epsilon\right\} .
\end{align*}
Since $S_{p_e}$ and $S_{\tilde{p}}$ have sizes that are at least $(1+\delta)L/2$ there exists a non-empty intersection. We have
$|S_{p_e}\cap S_{\tilde{p}}|=|S_{p_e}|+|S_{\tilde{p}}|-|S_{p_e}\cup S_{\tilde{p}}|\ge(1+\delta)L-L\ge\delta L\ge d+1$.

Since $p_e(x)$ and $\tilde{p}(x)$ are degree
$d$ polynomials by assumption, and $|p_e(x_i)-\tilde{p}(x_i)|\le|p_e(x_i)-y_i|+|y_i-\tilde{p}(x_i)|\le2|Q(x_i)|^2\epsilon$
for at least $d+1$ points in $x_0,\dotsc,x_{L-1}\in[-\Delta,\Delta]$. Then by Lemma~\ref{lem:Linterpolate} we obtain the
desired result 
\begin{align*}
&|p_e(1)-\tilde{p}(1)|\le\epsilon(2+O(m\Delta))\\
&\qquad\times\exp\left[d\left(1+\log\left((1+\Delta^{-1})\frac{L-1}{d}\right)\right)\right].
\end{align*}
Hence, $|p_e(1)-\mathcal{R}(C)|\le \epsilon(2+O(m\Delta))\allowbreak\exp\left[d(1+\log((1+\Delta^{-1})(L-1)/d))\right] + 2^{-\mathrm{poly}(m)}.$
Since $\Delta=\Theta(m^{-k})$ for some constant $k\ge 2$, we obtain Eq.~\eqref{eq:WeakLemma} in the Lemma.
\end{proof}
If we take $\delta=1/\poly(n)$, the success probability of $1/2+\delta$
can be boosted to a constant greater than $1/2$ by $O(\delta^{-2})$
calls to the oracle.

\section{\label{sec:CEQP}$\sharpP$-hardness of the computation of probability amplitude for the worst-case quantum circuit}
\begin{dfn}[$\mathsf{\sharpP}$~\cite{arora2009computational}]\label{def:sharpP}
A function $f\colon\{0,1\}^*\to\mathbb{N}$ is in $\sharpP$ if there exists a polynomial $p\colon\mathbb{N}\to\mathbb{N}$ and a polynomial-time deterministic
Turing machine $M$ such that for every $x\in\{0,1\}^*$,
\begin{equation*}
    f(x) = \left|\left\{y\in\{0,1\}^{p(|x|)}:\, M(x,y)=1\right\}\right|.
\end{equation*}
\end{dfn}
\begin{lem}[Equivalent to Thm 3.2 in Fenner et al \cite{fenner1998determining}]
\label{thm:Ampl_GapP}
For any $f\in\sharpP$ there is a polynomial-time uniform family of
quantum circuits $\{C_{\ell}(x)\}$ and a polynomial $p$ such that for
all $x$ of length $\ell$,
\begin{equation}\label{eq:Lem8}
|\langle0^{p(\ell)}|C_{\ell}(x)|0^{p(\ell)}\rangle|^{2}=\left(1-\frac{f(x)}{2^{p(\ell)-1}}\right)^2.
\end{equation}
\end{lem}
\begin{proof}
For any $f\in\sharpP$, there is a poly-time deterministic Turing machine $M(x,y)$ such that 
$f(x) = |\{y\in\{0,1\}^{p(\ell)}:\, M(x,y)=1\}|$
from Definition~\ref{def:sharpP} where $p(\ell)$ is a polynomial in $\ell=|x|$.
Let the poly-time uniform family of quantum circuits be $C_{\ell}(x)=H^{\otimes p(\ell)}V_{x}H^{\otimes p(\ell)}$
where $H$ is the Hadamard gate and
\[
V_{x}=\sum_{y\in\{0,1\}^{p(\ell)}}(-1)^{M(x,y)}|y\rangle\langle y|.
\]
Then
\begin{align*}
&\langle0^{p(\ell)}|C_{\ell}(x)|0^{p(\ell)}\rangle=\frac{1}{2^{p(\ell)}}\sum_{y\in\{0,1\}^{p(\ell)}}(-1)^{M(x,y)}\\
&=\frac{1}{2^{p(\ell)}}\left(\left|\left\{ y:\, M(x,y)=0\right\} \right|-\left|\left\{ y:\, M(x,y)=1\right\} \right|\right)\\
&=1-\frac{f(x)}{2^{p(\ell)-1}}.
\end{align*}
\end{proof}
\begin{lem}\label{lem:SharpP}
Let $\mathcal{O}$ be an oracle that for an arbitrary given quantum
circuit $C$ with $m$ gates computes $|\langle0^{n}|C|0^{n}\rangle|^{2}$ with
the additive error less than $2^{-m^{\mu}}$ for some constant $\mu>0$.
Then there exists an $\mathsf{FP}^{\mathcal{O}}$ algorithm that solves any $\sharpP$
problem.
\end{lem}
\begin{proof}
Let $f$ be an arbitrary $\sharpP$ function.
Let $M(x,y)$ be defined as before such that 
$f(x) = |\{y\in\{0,1\}^{p(|x|)}:\, M(x,y)=1\}|$. 

 Since in the two cases $f(x)=2^{p(|x|)-1}\pm c$ one obtains the same probability  (Eq.~\eqref{eq:Lem8}), we introduce a function $g$ that is sign unambiguous with respect to $1-\frac{g(x)}{2^{p(\ell)-1}}$ and is in one-to-one correspondence with the the probability amplitude. Let $g\colon \{0,1\}^*\to\mathbb{N}$ be a function defined by $g(x) = f(x) + 2^{p(|x|)}$.
Here, $G(x,y,z)\equiv z\vee M(x,y)$ satisfies $g(x)=|\{y\in\{0,1\}^{p(|x|)}, z\in\{0,1\} :\, G(x,y,z)=1\}|$.
Since $G(x,y,z)$ is computable in polynomial time, $g$ is also a $\sharpP$ function.
We will show an $\mathsf{FP}^{\mathcal{O}}$ algorithm that computes $g(x)$, and therefore it also computes $f(x)$.

As in Lemma~\ref{thm:Ampl_GapP}, there exists poly-time uniform family of quantum circuits $C_\ell(x)$ such that
$|\langle 0^{n}|C_\ell(x)|0^{n}\rangle|^2=(1-g(x)/2^{n-1})^2$ where $n\equiv p(|x|)+1$. The value of $g(x)$ is exactly determined from the approximation of $(1-g(x)/2^{n-1})^2$ to within the
additive error less than $1/2^{2n-1}$. Therefore, to compute $f(x)$, it is sufficient to compute $(1-g(x)/2^{n-1})^2$ with
additive error less than $1/2^{2n-1}$.

For computing $|\langle 0^n|C_\ell(x)|0^n\rangle|^2$ with small additive error with the aid of the oracle $\mathcal{O}$,
we use an enlarged quantum circuit $C_\ell(x) \cdot \mathbb{I}_1^k$, where $\mathbb{I}_1^k$ is $k$ identity gates acting on the first qubit (the number of qubits remains the same).
Since $|\langle 0^n|C_\ell(x)\cdot\mathbb{I}_1^k|0^n\rangle|^2 = |\langle 0^n|C_\ell(x)|0^n\rangle|^2$,  by assumption of the lemma
$\mathcal{O}(C_\ell(x)\cdot\mathbb{I}_1^k)$ outputs $(1-g(x)/2^{n-1})^2$ with additive error less than $2^{-(m+k)^\mu}$
where $m$ is the number of gates in $C_\ell(x)$.
By choosing $k=\lceil(2n)^{1/\mu}\rceil$, the additive error is upper bounded by $2^{-k^\mu}\le 2^{-2n}$.
\end{proof}
These prove our main theorems which we restate along with their proofs:
\setcounter{thm}{0}
\begin{thm}
There is an architecture $\mathcal{A}$ such that it is $\sharpP$-hard under $\mathsf{BPP}$-reduction
to approximate $|\braket{0|C|0}|^2$ with probability $1-O(1/m)$ over the choice of $C\sim\mathcal{H}_\mathcal{A}$  to within the additive error $2^{-\Omega(m \log m)}$.
\end{thm}
\begin{proof}
This is immediate from Lemmas \ref{lem:strong_oracle} and \ref{lem:SharpP}.
\end{proof}
\begin{thm}
There is an architecture $\mathcal{A}$ such that it is $\sharpP$-hard under $\mathsf{BPP}^\mathsf{NP}$-reduction
to approximate $|\braket{0|C|0}|^2$ with probability $\frac34 + \frac1{\poly(n)}$ over the choice of $C\sim\mathcal{H}_{\mathcal{A}}$ 
to within the additive error $2^{-\Omega(m \log m)}$.
\end{thm}
\begin{proof}
This is immediate from Lemmas \ref{lem:weak_oracle} and \ref{lem:SharpP}.
\end{proof}
\begin{rem}
Since $\sharpP$-hardness implies $\ceqp$-hardness with respect to Turing reduction, in the statements of the main theorems above one can simply replace $\sharpP$-hard with $\ceqp$-hard with respect to Turing reduction.
\end{rem}

\section{Hardness of fixed circuits}
In this section we show that any circuit with the same architecture as the worst-case circuit also has the same hardness properties as shown above. A perhaps surprising corollary is that sampling from circuits close to identity as $\# P$-hard. 
\begin{thm}
Let $U=U_m\cdots U_1$ be a fixed quantum circuit with an architecture $\mathcal{A}$. Let $\mathcal{O}$ be an oracle approximating the output probability  to within the  additive error $\exp(-\Omega(m\log m))$ of a circuit with the architecture $\mathcal{A}$ whose local gates are $\Delta$-close to $U_k$ in spectral norm, where $\Delta=o(m^{-1})$. Then, there is a $\mathsf{BPP}^{\mathsf{NP},\mathcal{O}}$ algorithm solving \#P-problem.
\end{thm}
\begin{proof}
Let $U_k(\theta)=U_k f(\theta h_k)$, where $f(h_k)=U^\dagger_k C_k$ and $C_k$ is the corresponding gate of the worst-case circuit. We have $U(0)=U_k$ and $U_k(1)=C_k$. Let us quantify the distance between the distribution over $U_k(\theta)$ and $U_k$ for $|\theta|\le\Delta$.

Let $\{\lambda_j\}_j$ be eigenvalues of the Hermitian matrix $h_k$.
Using the spectral norm and its invariance under unitary multiplication we have
\begin{align*}
&\| U_k(0)-U_k(\theta)\|_\infty = \| f_k(0)-f_k(\theta h_k)\|_\infty\\
&=\max_j \left|1-\frac{1+i\theta\lambda_j}{1-i\theta\lambda_j}\right|
=2|\theta|\max_j \left|\frac{\lambda_j}{1-i\theta\lambda_j}\right|\\
&\le 2\Delta\max_j \left|\lambda_j\right|=2\Delta \|h_k\|_\infty= O(\Delta)\;,
\end{align*}
where we used $\|h_k\|_\infty = O(1)$ because of the following argument.  Now $\|h\|_\infty$ is large if the unitary matrix $U_k^\dagger C_k$ has an eigenvalue near $-1$.  We have the freedom to multiply each $C_k$ by a global phase $C_k\rightarrow e^{i\varphi_k} C_k$, as this will not change the observables. In particular, all probabilities and $\# P$-hardness of the worst-case circuit $C$ are invariant under this transformation.
The multiplication of the global phase can be used to avoid the eigenvalue near $-1$.
Here, this unitary matrix has size at most four, and we can always multiply $C_k$ by an appropriate $e^{i\varphi_k}$ to rotate any eigenvalues close to $-1$ such that all eigenvalues of $U_k^\dagger C_k$ are a constant distance away from $-1$ rendering $\max_j \left|\lambda_j\right|=\|h_k\|_\infty\le f^{-1}(\exp(i 3\pi/4))=O(1)$.
\end{proof}
This theorem leads to the following result:
\begin{cor}\label{cor:id}
In Theorem 3 the local gates can be identity $U_k=I$. 
\end{cor}
Gottesman-Knill theorem states that circuits made up of Clifford gates are classically efficient to simulate. The Sum-over-Cliffords method of~\cite{bravyi2019simulation}\text[section 2.3] gives classical algorithms for sampling the output of near-Clifford circuits with non-trivial improvement on the exponential scaling of the run-time. However the above corollary shows that approximating the probability amplitudes of {\it near} Clifford circuits to within the stated additive error cannot be performed efficiently on a classical computer unless the polynomial hierarchy collapses to finite level.

Lastly, we comment on a class of quantum circuits known as IQP circuit~\cite{bremner2011classical,bremner2016achieving,bremner2016average}.  These circuits have the following form $C=H^{\otimes n}C_Z H^{\otimes n}$-- that is, the first and last layers are $n$ Hadamard gates and all intermediate gates are diagonal in the $Z$-basis. We believe our techniques along with the proof of the total variation distance in~\cite{movassagh2020quantum} can be generalized to random IQP circuits in a straightforward manner. This would imply the hardness of approximating the output probabilities of average-case IQP circuits.\\

\subsubsection*{Acknowledgement}
R.~Movassagh acknowledges funding from the MIT-IBM Watson AI Lab under the project {\it Machine Learning in Hilbert space}. The research was supported by the IBM Research Frontiers Institute.
R.~Mori was supported in part by JST PRESTO Grant Number JPMJPR1867
and JSPS KAKENHI Grant Numbers JP17K17711, JP18H04090, JP20H04138, and JP20H05966.

\bibliographystyle{IEEEtran}
\bibliography{mybib}

\begin{thebibliography}{10}
\providecommand{\url}[1]{#1}
\csname url@samestyle\endcsname
\providecommand{\newblock}{\relax}
\providecommand{\bibinfo}[2]{#2}
\providecommand{\BIBentrySTDinterwordspacing}{\spaceskip=0pt\relax}
\providecommand{\BIBentryALTinterwordstretchfactor}{4}
\providecommand{\BIBentryALTinterwordspacing}{\spaceskip=\fontdimen2\font plus
\BIBentryALTinterwordstretchfactor\fontdimen3\font minus
  \fontdimen4\font\relax}
\providecommand{\BIBforeignlanguage}[2]{{%
\expandafter\ifx\csname l@#1\endcsname\relax
\typeout{** WARNING: IEEEtran.bst: No hyphenation pattern has been}%
\typeout{** loaded for the language `#1'. Using the pattern for}%
\typeout{** the default language instead.}%
\else
\language=\csname l@#1\endcsname
\fi
#2}}
\providecommand{\BIBdecl}{\relax}
\BIBdecl

\bibitem{arute2019quantum}
F.~Arute, K.~Arya, R.~Babbush, D.~Bacon, J.~C. Bardin, R.~Barends, R.~Biswas,
  S.~Boixo, F.~G. Brandao, D.~A. Buell \emph{et~al.}, ``Quantum supremacy using
  a programmable superconducting processor,'' \emph{Nature}, vol. 574, no.
  7779, pp. 505--510, 2019.

\bibitem{grover1996fast}
L.~K. Grover, ``A fast quantum mechanical algorithm for database search,'' in
  \emph{Proceedings of the twenty-eighth annual ACM symposium on Theory of
  computing}, 1996, pp. 212--219.

\bibitem{simon1997power}
D.~R. Simon, ``On the power of quantum computation,'' \emph{SIAM journal on
  computing}, vol.~26, no.~5, pp. 1474--1483, 1997.

\bibitem{shor1999polynomial}
P.~W. Shor, ``Polynomial-time algorithms for prime factorization and discrete
  logarithms on a quantum computer,'' \emph{SIAM review}, vol.~41, no.~2, pp.
  303--332, 1999.

\bibitem{feynman1986quantum}
R.~P. Feynman, ``Quantum mechanical computers,'' \emph{Foundations of physics},
  vol.~16, no.~6, pp. 507--531, 1986.

\bibitem{bremner2011classical}
M.~J. Bremner, R.~Jozsa, and D.~J. Shepherd, ``Classical simulation of
  commuting quantum computations implies collapse of the polynomial
  hierarchy,'' in \emph{Proceedings of the Royal Society of London A:
  Mathematical, Physical and Engineering Sciences}, vol. 467, no. 2126.\hskip
  1em plus 0.5em minus 0.4em\relax The Royal Society, 2011, pp. 459--472.

\bibitem{terhal2002adaptive}
B.~M. Terhal and D.~P. DiVincenzo, ``Adaptive quantum computation, constant
  depth quantum circuits and {Arthur--Merlin} games,'' \emph{Quant. Inf.
  Comp.}, vol.~4, no.~2, pp. 134--145, 2004.

\bibitem{fenner1998determining}
S.~Fenner, F.~Green, S.~Homer, and R.~Pruim, ``Determining acceptance
  possibility for a quantum computation is hard for the polynomial hierarchy,''
  \emph{arXiv preprint quant-ph/9812056}, 1998.

\bibitem{aaronson2011computational}
S.~Aaronson and A.~Arkhipov, ``The computational complexity of linear optics,''
  in \emph{Proceedings of the forty-third annual ACM symposium on Theory of
  computing}.\hskip 1em plus 0.5em minus 0.4em\relax ACM, 2011, pp. 333--342.

\bibitem{bremner2016average}
M.~J. Bremner, A.~Montanaro, and D.~J. Shepherd, ``Average-case complexity
  versus approximate simulation of commuting quantum computations,''
  \emph{Physical review letters}, vol. 117, no.~8, p. 080501, 2016.

\bibitem{boixo2018characterizing}
S.~Boixo, S.~V. Isakov, V.~N. Smelyanskiy, R.~Babbush, N.~Ding, Z.~Jiang, M.~J.
  Bremner, J.~M. Martinis, and H.~Neven, ``Characterizing quantum supremacy in
  near-term devices,'' \emph{Nature Physics}, vol.~14, no.~6, p. 595, 2018.

\bibitem{napp2019efficient}
J.~Napp, R.~L. La~Placa, A.~M. Dalzell, F.~G. Brandao, and A.~W. Harrow,
  ``Efficient classical simulation of random shallow {2D} quantum circuits,''
  \emph{arXiv preprint arXiv:2001.00021}, 2020.

\bibitem{huang2020classical}
C.~Huang, F.~Zhang, M.~Newman, J.~Cai, X.~Gao, Z.~Tian, J.~Wu, H.~Xu, H.~Yu,
  B.~Yuan \emph{et~al.}, ``Classical simulation of quantum supremacy
  circuits,'' \emph{arXiv preprint arXiv:2005.06787}, 2020.

\bibitem{pednault2019leveraging}
E.~Pednault, J.~A. Gunnels, G.~Nannicini, L.~Horesh, and R.~Wisnieff,
  ``Leveraging secondary storage to simulate deep 54-qubit sycamore circuits,''
  \emph{arXiv preprint arXiv:1910.09534}, 2019.

\bibitem{harrow2018approximate}
A.~Harrow and S.~Mehraban, ``Approximate unitary $ t $-designs by short random
  quantum circuits using nearest-neighbor and long-range gates,'' \emph{arXiv
  preprint arXiv:1809.06957}, 2018.

\bibitem{dalzell2020random}
A.~M. Dalzell, N.~Hunter-Jones, and F.~G. Brand{\~a}o, ``Random quantum
  circuits anti-concentrate in log depth,'' \emph{arXiv preprint
  arXiv:2011.12277}, 2020.

\bibitem{stockmeyer1985approximation}
L.~Stockmeyer, ``On approximation algorithms for $\#\mathsf{P}$,'' \emph{SIAM
  Journal on Computing}, vol.~14, no.~4, pp. 849--861, 1985.

\bibitem{bouland2018quantum}
A.~Bouland, B.~Fefferman, C.~Nirkhe, and U.~Vazirani, ``On the complexity and
  verification of quantum random circuit sampling,'' \emph{Nature Physics},
  vol.~15, no.~2, p. 159, 2019.

\bibitem{movassagh2020quantum}
R.~Movassagh, ``Quantum supremacy and random circuits,'' \emph{arXiv preprint
  arXiv:1909.06210}, 2020.

\bibitem{movassagh2018efficient}
------, ``Efficient unitary paths and quantum computational supremacy: A proof
  of average-case hardness of random circuit sampling,'' \emph{arXiv preprint
  arXiv:1810.04681}, 2018.

\bibitem{paturi1992degree}
R.~Paturi, ``On the degree of polynomials that approximate symmetric boolean
  functions (preliminary version),'' in \emph{Proceedings of the twenty-fourth
  annual ACM symposium on Theory of computing}.\hskip 1em plus 0.5em minus
  0.4em\relax ACM, 1992, pp. 468--474.

\bibitem{rakhmanov2007bounds}
E.~A. Rakhmanov, ``Bounds for polynomials with a unit discrete norm,''
  \emph{Annals of mathematics}, pp. 55--88, 2007.

\bibitem{welch1986error}
L.~R. Welch and E.~R. Berlekamp, ``Error correction for algebraic block
  codes,'' Dec.~30 1986, uS Patent 4 633 470.

\bibitem{bouland2021noise}
A.~Bouland, B.~Fefferman, Z.~Landau, and Y.~Liu, ``Noise and the frontier of
  quantum supremacy,'' \emph{arXiv preprint arXiv:2102.01738}, 2021.

\bibitem{KondoBoson2021}
Y.~Kondo, R.~Mori, and R.~Movassagh, ``Fine-grained analysis and improved
  robustness of quantum supremacy for bosonsampling,'' \emph{in preparation},
  2021.

\bibitem{zhong2020quantum}
H.-S. Zhong, H.~Wang, Y.-H. Deng, M.-C. Chen, L.-C. Peng, Y.-H. Luo, J.~Qin,
  D.~Wu, X.~Ding, Y.~Hu \emph{et~al.}, ``Quantum computational advantage using
  photons,'' \emph{Science}, vol. 370, no. 6523, pp. 1460--1463, 2020.

\bibitem{oszmaniec2020fermion}
M.~Oszmaniec, N.~Dangniam, M.~E. Morales, and Z.~Zimbor{\'a}s, ``Fermion
  sampling: a robust quantum computational advantage scheme using fermionic
  linear optics and magic input states,'' \emph{arXiv preprint
  arXiv:2012.15825}, 2020.

\bibitem{arora2009computational}
S.~Arora and B.~Barak, \emph{Computational Complexity: A Modern Approach},
  1st~ed.\hskip 1em plus 0.5em minus 0.4em\relax USA: Cambridge University
  Press, 2009.

\bibitem{bravyi2019simulation}
S.~Bravyi, D.~Browne, P.~Calpin, E.~Campbell, D.~Gosset, and M.~Howard,
  ``Simulation of quantum circuits by low-rank stabilizer decompositions,''
  \emph{Quantum}, vol.~3, p. 181, 2019.

\bibitem{bremner2016achieving}
M.~J. Bremner, A.~Montanaro, and D.~J. Shepherd, ``Achieving quantum supremacy
  with sparse and noisy commuting quantum computations,'' \emph{arXiv preprint
  arXiv:1610.01808}, 2016.

\end{thebibliography}

\appendices
\section{Proof of Lemma~\ref{lem:Linterpolate}}\label{apx:Linterpolate}
Let $p_{j}=p(x_{j})$ for all $j=\{0,1,2,\dots,d\}$, where by assumption
$|p_{j}|\le\epsilon$. The Lagrange representation of the function $p(x)$
writes
\[
p(x)=\sum_{j=0}^{d}p_{j}\;\delta_{j}(x),\qquad\delta_{j}(x)\equiv\frac{\prod_{\ell\ne j}x-x_{\ell}}{\prod_{\ell\ne j}x_{j}-x_{\ell}}\;.
\]
By triangular inequality we have $|p(1)|\le\epsilon\sum_{j=0}^{d}|\delta_{j}(1)|$.
Moreover, using the fact that $|x_{j}|\leq \Delta$ for
all $j$ we have
\begin{align*}
|\delta_{j}(1)|&=\frac{\prod_{\ell\ne j}|1-x_{\ell}|}{\prod_{\ell\ne j}|x_{j}-x_{\ell}|}<\frac{\prod_{\ell\ne j}(1+\Delta)}{\prod_{\ell\ne j}|x_{j}-x_{\ell}|}\\
&<\frac{(1+\Delta)^d}{\prod_{\ell\ne j}|x_{j}-x_{\ell}|}.
\end{align*}
Since $|x_{j}-x_{\ell}|= \frac{2\Delta}{L-1}|a_j-a_\ell|\geq \frac{2\Delta}{L-1}|j-\ell|$, we have $\prod_{\ell\ne j}|x_{j}-x_{\ell}|\geq (\frac{2\Delta}{L-1})^{d}\prod_{\ell\ne j}|j-\ell|$
. Moreover $\prod_{\ell\ne j}|j-\ell|=\prod_{\ell\in\{0,1,2,\dots,,j-1,j+1,\dots,d\}}|j-\ell|=j!(d-j)!$
and we obtain
\begin{align*}
|\delta_{j}(1)|&< \left(\frac{L-1}{2\Delta}\right)^{d}\frac{(1+\Delta)^d}{\prod_{\ell\ne j}|j-\ell|}\\
&=\left(\frac{L-1}{2\Delta}\right)^{d}\frac{(1+\Delta)^d}{j!\:(d-j)!}\;.
\end{align*}
We express the bound  $|p(1)|\le\epsilon\sum_{j=0}^{d}|\delta_{j}(1)|$
as
\begin{align*}
|p(1)| & < \epsilon\left(\frac{(L-1)(1+\Delta)}{2\Delta}\right)^{d}\sum_{j=0}^{d}\frac{1}{j!\:(d-j)!}\\
&=\epsilon\left(\frac{(L-1)(1+\Delta)}{2\Delta}\right)^{d}\frac{1}{d!}\sum_{j=0}^{d}\binom{d}{j}\\
&=\epsilon\left(\frac{(L-1)(1+\Delta)}{2\Delta}\right)^{d}\frac{2^d}{d!}.
\end{align*}
By Stirling's inequality $n! \ge \sqrt{2\pi n}\frac{n^n}{\mathrm{e}^n}$, we conclude that
\begin{align*}
|p(1)|&<\epsilon\,\left(\frac{(L-1)(1+\Delta)}{2\Delta}\right)^{d}\frac{2^d}{d!}\\
&\le
\epsilon \frac{\left(\mathrm{e}(L-1)(1+\Delta^{-1})/d\right)^{d} }{\sqrt{2\pi d}}\\
&= \epsilon \,  \frac{\exp [d (1 + \log ((1+\Delta^{-1})\frac{L-1}{d}))]}{\sqrt{2\pi d}}.
\end{align*}

\section{The proof of the concentration}\label{apx:conc}
The proof was originally given in \cite{aaronson2011computational}, see also Theorem 1 in \cite{movassagh2020quantum}.
From Lemma~\ref{lem:tvd},
\[
\Pr_{C\sim\mathcal{H}_{\mathcal{A},\Delta}}\left[|\mathcal{O}(C) - p(\Delta)|\le|Q(\Delta)|^2\epsilon\right] \ge \frac34 + \delta - O(m\Delta)
\]
Let 
\[
\Theta\defeq\{i\in\{0,\dotsc,L-1\}\mid |p(x_{i})-y_{i}|\le |Q(x_i)|^2\epsilon\}.
\]
Then
\begin{align*}
&\Pr_{C\sim\mathcal{H}_{\mathcal{A},\Delta}}\left[\left|\Theta\right|
\ge (1+\delta)L/2\right]\\
&= 1-\Pr_{C\sim\mathcal{H}_{\mathcal{A},\Delta}}\left[\left|\Theta\right|< (1+\delta)L/2\right]\\
&= 1-\Pr_{C\sim\mathcal{H}_{\mathcal{A},\Delta}}\left[L-\left|\Theta\right|> (1-\delta)L/2\right]\\
&\ge 1 - \frac{(\frac14-\delta+O(m\Delta))L}{(1-\delta)L/2}\\
&= \frac12 + \frac32\frac{\delta}{1-\delta} - O(m\Delta).
\end{align*}
The inequality is obtained from Markov's inequality.

\end{document}